\documentclass{article}

\usepackage[utf8]{inputenc}
\usepackage[T1]{fontenc}
\usepackage[english]{babel}
\usepackage{amsmath,amssymb,amsthm}
\usepackage{graphicx} 
\usepackage[margin=1.in]{geometry}
\usepackage[shortlabels]{enumitem}
\usepackage{tikz,tikz-cd}
\usepackage{mathtools}
\usepackage{apptools}\AtAppendix{\counterwithin{theorem}{section}\counterwithin{equation}{section}}
\usepackage{todonotes}
\usepackage{csquotes}\MakeOuterQuote"
\usepackage{soul}\setstcolor{red}
\usepackage[colorlinks,citecolor=blue,linkcolor=blue]{hyperref}
\usepackage{graphicx} 
\usepackage{orcidlink}
\usepackage[bibstyle=numeric, backend=biber, sorting=none, citestyle=numeric-comp, giveninits,url=false,maxbibnames=5,isbn=false,date=year]{biblatex} 
\usepackage{cleveref}

\newtheorem{theorem}{Theorem}
\newtheorem{lemma}[theorem]{Lemma}
\newtheorem{corollary}[theorem]{Corollary}
\newtheorem{proposition}[theorem]{Proposition}

\newtheorem{definition}[theorem]{Definition}

\theoremstyle{definition}

\newtheorem*{problem*}{Problem}
\newtheorem*{assumption*}{Assumption}
\newtheorem{example}[theorem]{Example}
\newtheorem{remark}[theorem]{Remark}
\newtheorem*{warning*}{Warning}

\crefname{theorem}{Thm.}{Thms.}
\Crefname{theorem}{Theorem}{Theorems}
\crefname{proposition}{Prop.}{Props.}
\Crefname{proposition}{Proposition}{Propositions}
\crefname{lemma}{Lem.}{Lemmas}
\Crefname{lemma}{Lemma}{Lemmas}
\crefname{corollary}{Cor.}{Cors.}
\Crefname{corollary}{Corollary}{Corollaries}
\crefname{definition}{Def.}{Defs.}
\Crefname{definition}{Definition}{Definitions}
\crefname{equation}{eq.}{eqs.}
\Crefname{equation}{Equation}{Equations}
\crefname{figure}{fig.}{figs.}
\Crefname{figure}{Figure}{Figures}
\crefname{remark}{Rem.}{Rems.}
\crefname{remark}{Remark}{Remarks}
\crefname{appendix}{Appendix}{Appendices}
\Crefname{appendix}{Appendix}{Appendices}

\newcommand{\braket}[2]{\langle #1|#2\rangle}

\newcommand{\ket}[1]{|#1\rangle}

\DeclareMathOperator{\supp}{supp}
\newcommand{\norm}[1]{\lVert #1\rVert}

\newcommand{\ox}{\otimes}
\newcommand{\mc}{\mathcal}
\newcommand{\eps}{\varepsilon}

\DeclareMathOperator{\tr}{Tr}
\DeclareMathOperator{\Tr}{Tr} 

\renewcommand{\tilde}{\widetilde}
\renewcommand{\hat}{\widehat}

\newcommand{\hide}[1]{}

\def\CC{{\mathbb C}}

\def\H{{\mc H}}

\def\M{{\mc M}}

\def\RR{{\mathbb R}}

\def\NN{{\mathbb N}}

\DeclareMathOperator{\spec}{Sp}

\newcommand{\CAR}{\textup{CAR}}

\newcommand*{\1}{\text{\usefont{U}{bbold}{m}{n}1}}

\renewcommand{\Re}{\mathrm{Re}\,}
\renewcommand{\Im}{\mathrm{Im}\,}

\newcommand{\dist}{\mathrm{dist}}
\newcommand{\ip}[2]{\left(#1,#2\right)}
\newcommand{\h}{\mathfrak h}
\newcommand{\n}{\mathfrak n}
\renewcommand{\r}{\mathfrak r}
\renewcommand{\k}{\mathfrak k}
\newcommand{\ptr}{\mathrel{\hat =}}

\usepackage{tikz}

\title{\vspace{-2cm} Gaussian fermionic embezzlement of entanglement}
\author{Alessia Kera\,\orcidlink{0009-0000-9664-1248}, Lauritz van Luijk\,\orcidlink{0000-0003-3153-549X}, Alexander Stottmeister\,\orcidlink{0000-0002-0145-0877}, Henrik Wilming\,\orcidlink{0000-0002-0306-7679}}
\date{\small Leibniz Universit\"at Hannover, Institut f\"ur Theoretische Physik, Appelstraße 2, 30167 Hannover, Germany}

\begin{document}
\maketitle
\begin{abstract}
Embezzlement of entanglement allows to extract arbitrary entangled states from a suitable \emph{embezzling state} using only local operations while perturbing the resource state arbitrarily little. 
A natural family of embezzling states is given by ground states of non-interacting, critical fermions in one spatial dimension. 
This raises the question of whether the embezzlement operations can be restricted to Gaussian operations whenever one only wishes to extract Gaussian entangled states. 
We show that this is indeed the case and prove that the embezzling property is in fact a generic property of fermionic Gaussian states. Our results provide a fine-grained understanding of embezzlement of entanglement for fermionic Gaussian states in the finite-size regime and thereby bridge finite-size systems to abstract characterizations based on the classification of von Neumann algebras.
To prove our results, we establish novel bounds relating the distance of covariances to the trace-distance of Gaussian states, which may be of independent interest.
\end{abstract}

\tableofcontents
\section{Introduction}
Just like macroscopic bulk properties in the thermodynamic limit of many-body systems, some properties of entanglement only sharply emerge in systems with an unbounded number of degrees of freedom. 
Recent research has shown that these large-scale entanglement properties of many-body systems are closely tied to the von Neumann algebraic classification of the emergent subsystem operator algebras \cite{van_luijk_large-scale_2025, long_paper, locc}.
Of particular interest are those large-scale properties that only sharply manifest in the thermodynamic limit while still implying strong operational consequences for finite-size systems.

A paradigmatic example for such an effect is \emph{embezzlement of entanglement} \cite{van_dam2003universal,short_paper}: Here, two agents acting on disjoint subsystems $A,B$ of the total system, can extract an arbitrary bipartite entangled state $\ket{\Psi}_{A'B'} \in \CC^d\ox \CC^d$ on a system $A'B'$ from an \emph{embezzling state} $\ket{\Omega}_{AB} \in \H$ 
to any error $\eps>0$ while also perturbing the embezzling state arbitrarily little:
\begin{align}
    u_{AA'} u_{BB'} \ket{\Omega}_{AB}\ox \ket{0}_{A'}\ket{0}_{B'} \approx_\eps \ket{\Omega}_{AB}\ox\ket{\Psi}_{A'B'},\nonumber
\end{align}
where $\ket{0}_{A'}\ket{0}_{B'}$ is an arbitrary product state on $A'B'$ and $u_{AA'}, u_{BB'}$ are unitaries on the respective subsystems.

Embezzling states require infinitely many degrees of freedom. Whether they exist is determined by the operator algebras describing subsystems $A$ and $B$---in particular, by their emergent von Neumann algebraic type and subtype \cite{long_paper,short_paper}. Yet, their embezzling properties descend to finite-size versions of the same system in an approximate form \cite{van_luijk_critical_2025,luijkMultipartiteEmbezzlementEntanglement2025}.  
There are even systems where all states in the Hilbert space have the embezzling property. 
It was shown in \cite{van_luijk_critical_2025} that ground state sectors of one-dimensional, critical, fermionic models with quadratic Hamiltonians fall into this class, which is consistent with the fact that the scaling limit of such systems should give rise to conformal field theories \cite{Sato1980, vidal_entanglement_2007, Vidal2008, evenbly_real-space_2014, witteveen_quantum_2022, osborne2023cft_lattice_fermions, osborne2023ising}. 
As a consequence, also systems with gapless, chiral edge modes described by non-interacting fermions, such as the $p_x + ip_y$ superconductor \cite{read_paired_2000,ivanov_non-abelian_2001}, have embezzling ground states.
It is plausible, but unproven, that all critical systems in one spatial dimension and all systems with a chiral edge in two spatial dimensions share this property. 
So far, the abstract operator-algebraic results say little about
\begin{itemize}
    \item[(i)] the unitary operations that need to be employed to implement the embezzlement task,
    \item[(ii)] how the achievable error $\eps$ scales with the system size when considering finite-size systems (see, however, \cite{schwartzmanComplexityEntanglementEmbezzlement2025a,van_luijk_critical_2025}),
    \item[(iii)] and how generic the (approximate) embezzling property is.
\end{itemize}

We considerably advance the situation in this paper for ground states of quadratic, fermionic Hamiltonians, which are Gaussian (also called quasi-free) and hence fully described by their two-point correlation functions. Our main result is as follows:
We precisely characterize when Gaussian fermionic states can act as (approximate) embezzling states under the restricted class of Gaussian unitary operations (those generated by quadratic fermionic Hamiltonians) and give precise error estimates. 
In particular, our result characterizes \emph{universal Gaussian embezzling families}, which allow for embezzling arbitrary Gaussian entangled states using only Gaussian local operations. Our results imply that such states are generic among the pure entangled Gaussian states.
Examples comprise the ground states of the $XX$-spin chain or the transverse-field Ising model on increasing system sizes, or the above-mentioned $p_x + ip_y$ superconductors in a suitable geometry. 
We also show that these families allow to embezzle arbitrary entangled states if we allow for general (non-Gaussian) local unitary operations.

\section{Gaussian formalism}
We briefly review the formalism of fermionic Gaussian states.
In the main text, we restrict to \emph{passive} (or \emph{gauge-invariant}) Gaussian states for simplicity. These are characterized by the fact that they commute with the number operator. 
However, all our results also hold for general Gaussian states, which are, for example, necessary to describe the transverse-field Ising model or superconductors. 
These more general results are discussed in the appendix.
Consider a finite number $n$ of fermionic modes, with associated creation and annihilation operators (CAOs) $a_i^\dagger, a_j$ as well as two subsets $I,J \subset \{1,\ldots,n\}$ with elements $i_k\in I, j_k \in J$. A passive Gaussian state is a density matrix $\rho_G$ on the Fock space that fulfills
\begin{align}
    \tr(\rho_G  a_{i_{|I|}} \cdots a_{i_1} a_{j_1}^\dagger \cdots a_{j_{|J|}}^\dagger) = \delta_{|I|,|J|} \det G_{IJ}, \label{eq:wick}
\end{align}
where $G \in \CC^{n\times n}$ is a positive semidefinite matrix fulfilling $0\leq G\leq \1$ and $G_{IJ}$ is the submatrix of entries $G_{i,j}$ with $i\in I,j\in J$. We call $G$ the covariance of $\rho$. 
Pure Gaussian states are represented by projections, i.e., fulfill $G^2=G$, and correspond to Fock states (states with a definite number of fermions) relative to a suitable single-particle basis.
Any passive Gaussian state on $n$ modes has a purification in terms of a pure passive Gaussian state on a system containing $2n$ modes.

A unitary operator $u$ on the underlying single-particle Hilbert space $\CC^n$ naturally lifts to a unitary operator $\Gamma(u)$ on the Fock space acting as $\Gamma(u) a_i^\dagger \Gamma(u)^{\dagger} = \sum_{k} u_{ki} a^\dagger_k$. Such operators are called (passive) Gaussian unitaries and map (passive) Gaussian states to (passive) Gaussian states. 
They commute with the total number operator $N = \sum_i a_i^\dagger a_i$ and therefore leave the probability distribution of the total number of fermions invariant. 

\section{Gaussian embezzlement of entanglement}
Consider a bipartite system with $n_A+n_B$ modes, where the first $n_A$ modes correspond to the first subsystem and the remaining $n_B$ modes to the second subsystem. It follows from \cref{eq:wick} that Gaussian \emph{product states} have covariances that are direct sums, i.e., of the form $G=G_A\oplus G_B$, with reduced states $\rho_{G_A}$ and $\rho_{G_B}$. 
Thus $\rho_{G_A\oplus G_B} = \rho_{G_A}\ox \rho_{G_B}$.
Now consider a pure state $\rho_{K}$ on the bipartite system with $n_A+n_B$ modes and let $\rho_{G},\rho_{F}$ be pure states on a second bipartite system with $d_{A'}+d_{B'}$ modes.
For a given error $\eps>0$ we wish to characterize when it is possible to find single-particle unitary operators $u_{AA'} \in \CC^{(n_A+d_{A'})\times(n_A+d_{A'})}$ and $u_{BB'} \in \CC^{(n_B+d_{B'})\times(n_B+d_{B'})}$ such that
\begin{align}\label{eq:gaussian-mbz-ent}
	\dist\big(\Gamma(u)\rho_{K}\ox \rho_{F}\Gamma(u)^\dagger, \rho_K \ox \rho_G\big) \leq \eps,
\end{align}
where 
\begin{align}
    \dist(X,Y) = \frac{1}{2}\norm{X-Y}_1
\end{align} 
is the trace-distance and $u = u_{AA'}\oplus u_{BB'}$.
Let $K_A,G_{A'},F_{A'}$ denote the submatrices corresponding to the respective subsystems. 
Then by the monotonicity of the trace-distance under partial trace, \eqref{eq:gaussian-mbz-ent} implies
\begin{align}\label{eq:gaussian-mbz}
	\dist\big(\Gamma(u_{AA'})\rho_{K_{\!A}}\!\!\ox\!\rho_{F_{\!A'}}\Gamma(u_{AA'})^\dagger, \rho_{K_{\!A}}\!\!\ox\!\rho_{G_{\!A'}}\big)\!\leq\!\eps.
\end{align}
Conversely, we show in the appendix that if the latter equation is fulfilled, then \eqref{eq:gaussian-mbz-ent} is fulfilled with error at most $4 \eps^{1/2}$.
Thus, we can restrict to \eqref{eq:gaussian-mbz}, which describes the conversion of mixed Gaussian states using global unitary operations instead of the conversion of pure entangled states using local unitary operations. 
Dropping the labels $A$ and $A'$, for any triple of covariances $K\in \CC^{n\times n}, G,F\in \CC^{d\times d}$ we consider the function
\begin{align}
	\kappa(F,G|K) = \inf_{u} \dist\big(\Gamma(u) \rho_{K\oplus F} \Gamma(u)^\dagger, \rho_{K\oplus G}\big),
\end{align}
where  the infimum runs over the unitary group of the single-particle Hilbert space of the combined system.

\section{Main result}
We now present and discuss our main result. 
We say that a covariance $K$ has $\eps$-dense spectrum if for every $x\in[0,1]$ we have $|\lambda_j(K)-x|<\eps$ for some eigenvalue $\lambda_j(K)$ of $K$. Our main result shows that an $\eps$-dense spectrum is sufficient for approximate Gaussian embezzlement.

\begin{theorem}\label{thm:main} Let $K$ have $\eps$-dense spectrum and let $F,G$ be $d$-dimensional. Then
    \begin{align}
	    \kappa(F,G|K) &\leq 11\, d\, \eps^{1/4}. 
    \end{align} 
\end{theorem}

Let us discuss how to apply the result in concrete examples. Any covariance of a system with $n$ modes can have at best $1/n$-dense spectrum. 
Thus, optimal scaling of our bound is achieved if, for example, $K$ has eigenvalues $\lambda_j(K) = 1 - \frac{j}{n}$ with $j=1,\ldots,n$.
Then $\eps = 1/n$ and we find $\kappa(F,G|K) \leq 11 d n^{-1/4}$. Hence, as $n\rightarrow \infty$, any Gaussian state may be embezzled to arbitrary accuracy. 
We show in the appendix that a Gaussian state $\rho_K$ with $\eps$-dense spectrum has entropy  $H_1(\rho_K) \geq \log(2)/(4\eps)$, yielding extensive entropy. In the original bipartite setting, these entropies have to be interpreted as entanglement entropies. 
In non-interacting fermionic systems, extensive entanglement entropies of the ground states are incompatible with a local, translation-invariant Hamiltonian \cite{fannes_entropy_2003, wolf_violation_2006,gioev_entanglement_2006,farkas_sharpness_2005,farkas_von_2007}. Therefore, this family of states does not appear naturally in the context of many-body physics.

Consider now a critical model, corresponding to a scaling of Rényi entropies as
\begin{align}\label{eq:critical-scaling}
    H_\alpha(\rho_K) \approx \frac{c}{6}\Big(1+\frac{1}{\alpha}\Big)\log(n) + c_\alpha,\quad n\to \infty,
\end{align}
where $c_\alpha$ is a non-universal function of $\alpha$ \cite{holzhey_geometric_1994, calabrese2004entanglement_entropy_qft, calabrese_entanglement_2008}.
The bound $H_1(\rho_K) \geq \log(2)/(4\eps)$ now implies $\eps \geq O(1/\log(n))$ for a chain of length $2n$. 
On the other hand, the results of \cite{van_luijk_critical_2025} show that the covariance has a continuous spectrum in the thermodynamic limit. Hence, the spectrum becomes arbitrarily dense as $n\to\infty$. Thus, our main result recovers the fact that ground states of non-interacting, critical fermionic systems are embezzling states, but the convergence is very slow.

\subsection{From quasi-free to general embezzlement of entanglement}
We now show that an approximate Gaussian embezzling state can be used to embezzle arbitrary entangled states once we allow for general unitary operations instead of Gaussian unitary operations.
The argument is simple: It has been shown in \cite{leung_characteristics_2014}, that any entangled state that can embezzle arbitrary entangled states with Schmidt rank $2$ (i.e., qubit states up to local isometries) with fidelity $\mathrm{F}$ can also embezzle arbitrary pure entangled states with Schmidt rank $m$ with fidelity $\mathrm{F}_m^2 \geq 1 - (\log_2(m)^2)(1-\mathrm{F}^2)$.
It thus suffices to show that it is possible to embezzle arbitrary pure entangled states on a pair of two-level systems, which in turn means that it is sufficient to be able to embezzle arbitrary mixed qubit states on a single system. However, up to local unitary freedom, Gaussian states of a single fermionic mode precisely correspond to arbitrary mixed qubit states, which completes the argument.\footnote{If we wish to adhere to the parity superselection rule, we may encode a qubit into two fermionic modes via $\sigma_X \cong (a_1^\dagger - a_1)(a_2^\dagger + a_2)$ and $\sigma_Z \cong 2a_1^\dagger a_1 - 1$.}

\subsection{Embezzlement and number conservation}
A curious property of embezzlement is that it can arbitrarily alter the local probability distribution of the total fermion number, despite the fact that passive Gaussian unitaries preserve the latter. 
This is possible because an embezzling system must contain many more modes than the system on which it embezzles and must have a very broad distribution for the number of fermions (otherwise it cannot have $\eps$-dense spectrum for small $\eps$). 
Thus, one may shift the mean value (or any other finite number of moments) of the distribution of the number of fermions on the embezzling system by a finite amount, while keeping the probability distribution almost unaffected as measured by total variation distance, see also \cite{lipka-bartosikCatalysisQuantumInformation2024} for a discussion of this point for other forms of embezzlement.

\section{Proof of main result} We now present the proof of our main result. Proofs of the lemmata can be found in the appendix. The proof strategy is as follows: We first reduce the problem to the level of covariances and then further reduce it to the commuting case. Finally, we solve the classical problem explicitly.

In the context of Gaussian fermionic states, it is standard to bound the trace-distance between Gaussian states using distance measures on covariances. 
This allows one to reduce the problem from dimension $2^{n+d}$ to $n+d$. 
Very recently, the bound 
\begin{align}\label{eq:bittel}
	\dist(\rho_F,\rho_G) \leq 2\, \dist(F,G)
\end{align}
was established and shown to be optimal under certain circumstances \cite{bittel_optimal_2024}.
Unfortunately, \eqref{eq:bittel} is insufficient for our purposes. 
Indeed, we prove in the appendix that
\begin{align}
	\inf_u \dist\big(u(K\oplus F)u^\dagger, K\oplus G\big) \geq \inf_u \dist\big(uFu^\dagger,G\big).\nonumber 
\end{align}
The right-hand side cannot be made arbitrarily small for arbitrary $G,F$ of dimension $d$. A prototypical example is $F=\1$ (the fully occupied Fock state) and $G=0$ (the Fock vacuum), in which case the right-hand side evaluates to $d/2$.
Thus, when estimating the trace-distance of Gaussian states using the trace-distance of covariances, Gaussian embezzlement of entanglement seems to be impossible. 
We overcome this hurdle using a novel, alternative bound, which we consider to be of independent interest:
\begin{proposition}\label{prop:cov-bound}
    Let $A,B \in \CC^{n\times n}$ be covariances and let
    \begin{align}
         \eta(A,B) &:= \norm{\sqrt{\1-A}\sqrt{B} - \sqrt{A}\sqrt{\1-B}}_2.
    \end{align}
    Then
    \begin{align}
    	1 - e^{-\eta(A,B)^2/2} \leq  \dist(\rho_A,\rho_B) \leq  \sqrt{2}\, \eta(A,B).
    \end{align}
\end{proposition}

The proof of this result, which generalizes to non-passive (non-gauge-invariant) Gaussian states as well as infinite-dimensional single-particle Hilbert spaces, is presented in the appendix.
$\eta$ enjoys the useful property
\begin{align}\label{eq:eta_prop}
    \eta(A\oplus C, B\oplus C) = \eta(A,B).
\end{align}

We now come to the proof of \cref{thm:main}. 
In the following, it will be convenient to approximate $F,G$ by covariances $F_\delta, G_\delta$ such that
\begin{align}
\delta \1 \leq F_\delta, G_\delta \leq (1-\delta)\1
\end{align}
by replacing all eigenvalues $\lambda_j> 1-\delta$ with $1-\delta$ and all eigenvalues $\lambda_j<\delta$ with $\delta$. Here $0< \delta \leq \frac{1}{2}$ is a free parameter that will be optimized later on. 
Using the triangle-inequality and \cref{prop:cov-bound}, it follows that
\begin{align}
	\kappa(F,G|K) \leq \kappa(F_\delta,G_\delta|K) + \sqrt{2} \big(\eta(F_\delta,F) + \eta(G_\delta, G)\big).\nonumber 
\end{align}

\begin{lemma}\label{lemma:PS-trick} 
    Let $A,B\geq 0$ be covariances. If $\delta \1 \leq A\leq (1-\delta)\1$, then
    \begin{align}
        \eta(A,B)^2 \leq \frac{4}{\delta}\norm{A-B}_2^2.
    \end{align}
    If also $\delta \1 \leq B\leq(1-\delta)\1$, then
     \begin{align}
        \eta(A,B)^2 \leq \frac{1}{\delta}\norm{A-B}_2^2.
    \end{align}
\end{lemma}
\cref{lemma:PS-trick} implies that $\eta(F,F_\delta) \leq 2\sqrt{d\delta}$ and similarly for $G,G_\delta$. Thus 
\begin{align}
	\kappa(F,G|K) &\leq \kappa(F_\delta,G_\delta|K) + \sqrt{2}\cdot 4 \sqrt{d\delta}.
\end{align}
Using again \cref{prop:cov-bound}, it remains to bound
\begin{align}
	\kappa(F_\delta,G_\delta|K) &\leq \inf_u \sqrt{2} \eta(u(K\oplus F_\delta)u^\dagger, K\oplus G_\delta).\nonumber
\end{align}
If $K$ has $\eps$-dense spectrum, there is a subspace $P_\eps \subset \CC^n$ of dimension $n_\eps \leq \frac{2}{\eps}$ such that $K P_\eps = P_\eps K =:K_\eps$ (we identify $P_\eps$ with its orthogonal projector) and such that $K_\eps$ also has $\eps$-dense spectrum. In particular, $\lambda_j(K_\eps) - \lambda_{j+1}(K_\eps) \leq 2\eps$ and $\kappa(F,G|K)\leq \kappa(F,G|K_\eps)$. We can thus write $K=\tilde R\oplus K_\eps$.
We further denote by $K_\eps|_\delta$ the restriction of $K_\eps$ to the subspace $P_{\eps,\delta}$ of $P_\eps$ on which $\delta \1 \leq K_\eps \leq (1-\delta)\1$. 
Similarly, we restrict the unitaries $u$ in the optimization to only act non-trivially on the subspace $P_{\eps,\delta}\oplus \CC^d\subset \CC^n\oplus \CC^d$ and write $K = R\oplus K_\eps|_\delta$. Then, using \eqref{eq:eta_prop} and \cref{lemma:PS-trick},
\begin{align}
  \inf_u \eta(u(K \oplus F_\delta)u^\dagger, K \oplus G_\delta)  \nonumber  &= \inf_u \eta(u(K_\eps|_\delta \oplus F_\delta)u^\dagger, K_\eps|_\delta \oplus G_\delta) \leq \frac{1}{\delta^{1/2}} \inf_u \norm{u(K_\eps|_\delta \oplus F_\delta)u^\dagger - K_\eps|_\delta \oplus G_\delta}_2\nonumber \\
  &= \frac{1}{\delta^{1/2}}\norm{(k\oplus f)^\downarrow - (k\oplus g)^\downarrow}_2, \label{eq:hn}\\
  &\leq \frac{\sqrt{n_\eps + d}}{\delta^{1/2}}\norm{(k\oplus f)^\downarrow - (k\oplus g)^\downarrow}_\infty,
\end{align}
where $k,f,g$ denote the vectors of eigenvalues (repeated according to multiplicity) of $K_\eps|_\delta, F_\delta, G_\delta$, respectively,  and $v^\downarrow$ denotes non-increasing ordering of a vector $v\in \RR^s$, $s\in\NN$. 
Eq. \eqref{eq:hn} follows from
\begin{align}
    \inf_u \norm{u A u^\dagger- B}_p = \norm{\lambda(A)^\downarrow-\lambda(B)^\downarrow}_p,
\end{align}
which is a consequence of Wielandt's theorem (see, e.g., \cite[Theorem 9.G.1.a]{marshallInequalitiesTheoryMajorization2011}),
and the last inequality follows from $\norm{v}_2 \leq \sqrt{d}\norm{v}_\infty$ for any $d$-dimensional vector $v$.
We now make use of the following Lemma.
\begin{lemma}\label{lemma:list-sort}
     Let $\tilde k\in \RR^n$ and $g,f\in \RR^d$ be non-increasingly ordered vectors with $\tilde k_1\geq g_1,f_1$ and $\tilde k_n\leq f_d,g_d$. Then
     \begin{align}
         \norm{(\tilde k\oplus f)^\downarrow - (\tilde k\oplus g)^\downarrow}_\infty &\leq \max_{|i-j|\leq d}\, |\tilde k_i - \tilde k_j| \leq d\, \max_i (\tilde k_i - \tilde k_{i+1}).
     \end{align}
 \end{lemma}
In our case, it may happen that the conditions $k_1 \geq g_1,f_1$ and $k_{n_\eps} \leq g_d,f_d$ are not fulfilled. 
In this case we necessarily have $\tilde k_1 := \max \{k_1 + \eps ,1\} > g_1,f_1$ and $\tilde k_{n_\eps + 2} := \min \{k_{n_\eps} -\eps, 0\} < g_1,f_1$. 
We can therefore consider the vector $\tilde k = (\tilde k_{1}, k_1,\ldots,k_{n_\eps}, \tilde k_{n_\eps +2 })$ that fulfills
\begin{align}
    \norm{(\tilde k\oplus f)^\downarrow - (\tilde k\oplus g)^\downarrow}_\infty = \norm{(k\oplus f)^\downarrow - (k\oplus g)^\downarrow}_\infty\nonumber
\end{align}
and $\max_i (\tilde k_i - \tilde k_{i+1}) \leq 2\eps$. 
We thus find
\begin{align}
    \inf_u \eta&(u(K_\eps \oplus F_\delta)u^\dagger, K_\eps \oplus G_\delta)  \leq 2\frac{(n_\eps +d)^{1/2}\, d\, \eps }{\delta^{1/2}}.\nonumber
\end{align}
Putting together the derived bounds, we find
\begin{align}
	\kappa(F,G|K) \leq 2\sqrt{2}\sqrt{d \delta }\left(2 + \sqrt{(n_\eps + d)d}\frac{\eps}{\delta} \right).
\end{align}
Inserting $n_\eps\leq 2/\eps$ and choosing $\delta=\frac{1}{2}\big(\frac{2d}{\eps}+d^2\big)^{1/2}   \eps$ yields (using $2\eps + \eps^2 d \leq 3 \eps d$)
\begin{align}
    \kappa(F,G|K) \leq 8\big(d^3(2\eps+\eps^2d)\big)^{1/4} \leq 11\,  d\, \eps^{1/4}.\nonumber
\end{align}

\section{Conclusion}
We have shown that Gaussian fermionic entangled states with dense Gaussian entanglement spectrum can be used to embezzle arbitrary Gaussian pure entangled states using only Gaussian local operations and general pure entangled states using general local unitary operations.

This result provides a finite-size explanation of the recent discovery that all one-dimensional critical, non-interacting fermionic systems have embezzling ground states. 
In the thermodynamic limit, the results of \cite{van_luijk_critical_2025,long_paper} show that all quasi-local excitations of the ground state share the same embezzling properties. It would be desirable to show explicitly in the setting of quasi-free systems that the $\eps$-dense spectrum of the covariance is stable against Gaussian unitary operations that act on a subsystem of fixed size, but may couple the modes of Alice with those of Bob. We leave this open for future work. 

\vspace{-0.1cm}

\paragraph{Acknowledgments}
AS and LvL have been funded by a Stay Inspired Grant of the MWK Lower Saxony (Grant ID:
15-76251-2-Stay-9/22-16583/2022)

\printbibliography
\clearpage
\appendix
\noindent {\huge \bfseries Appendix}\\

\noindent We identify $1$ with the identity operator on a Hilbert space, but write $\1$ for its matrix representation in a given basis and $\1_n$ for the $n\times n$ unit matrix.

\section{The general Gaussian formalism and the distance bound}
Here we present, discuss, and prove our main result for the general case of non-passive (non-gauge-invariant) Gaussian fermionic states. 
There are various, equivalent, but slightly different formalisms to describe these states. 
One is Araki's self-dual formalism, describing a system of $N$ fermionic modes using a \emph{complex} $2N$-dimensional Hilbert space \cite{araki1970car_bog}. The other is the Majorana formalism, describing the same system using a \emph{real} $2N$-dimensional Hilbert space, see for example \cite{bittel_optimal_2024}. While the former is more common in mathematical physics, the latter is more common in quantum information theory.
We will therefore discuss the results in both formalisms in parallel. 
In contrast to the main text, we here take an operator algebraic perspective, where a quantum state is a positive, normalized functional on an operator algebra. Since we work with fermionic systems, the relevant operator algebra is the C* algebra generated by the canonical anti-commutation relations. 
Let $\h$ be a complex single-particle Hilbert space with inner product $\ip{\cdot}{\cdot}_\h$. 
We do not assume $\h$ to be finite dimensional, unless it is explicitly stated.
The associated algebra of canonical anti-commutation relations $\CAR(\h)$ is generated by the creation and annihilation operators (CAOs) $a,a^\dagger: \h\to \CAR(\h)$ fulfilling the commutation relations
\begin{align}
  \{a^\dagger(\phi),a^\dagger(\psi)\} = 0,\quad \{a(\phi),a^\dagger(\psi)\} = \ip{\phi}{\psi}_\h,
\end{align}
with $a^\dagger$ linear and $a$ conjugate-linear.
A fermionic quantum state is a linear map $\omega : \CAR(\h) \to \CC$ such that $\omega(x^\dagger x)\geq 0$ for all $x\in \CAR(\h)$ and $\omega(1) = 1$. 
A state is called even if it vanishes on odd polynomials in the CAOs.
If $\omega_1$ is an even state on $\CAR(\h)$ and $\omega_2$ is an even state on $\CAR(\k)$, there is a unique (even) state $\omega_1\ox\omega_2$ on $\CAR(\h\oplus \k)$ such that\footnote{We identify $a^\dagger(\psi)\in \CAR(\h)$ with $a^\dagger(\psi\oplus 0)\in \CAR(\h\oplus \k)$ and similarly for $\k$.}
\begin{align}
    \omega_1\ox\omega_2(xy) = \omega_1(x)\omega_2(y),\quad x\in \CAR(\h), y\in \CAR(\k). 
\end{align}
We now introduce the two formalisms that we will use later to describe Gaussian (quasi-free) states.

\subsection{Self-dual formalism}
In Araki's self-dual formalism, $\CAR(\h)$ is represented by combining the CAOs into field operators relative to a doubled single-particle Hilbert space $\h\oplus \h$:
\begin{align}
    B(\phi_1\oplus\phi_2)\equiv B(\phi_1,\phi_2) = a^\dagger(\phi_1) + a(J\phi_2),
\end{align}
where $J$ is an antiunitary involution on $\h$, i.e. $J^2=1$ and $\ip{J\psi}{J\phi} = \ip{\phi}{\psi}$. If $\h=\CC^m$, $J$ may be chosen to be the complex conjugation relative to the standard basis.

Let $C$ be the anti-linear involution on $\h\oplus \h$ defined by $C \phi_1 \oplus \phi_2 = J\phi_2 \oplus J \phi_1$. Then $B(\phi_1,\phi_2)^\dagger = B(C \phi_1\oplus \phi_2)$ and we find
\begin{align}
    B(\phi_1\oplus\phi_2)B(\psi_1\oplus\psi_2)^\dagger + B(\psi_1\oplus \psi_2)^\dagger B(\phi_1\oplus\phi_2) = \ip{\psi_1\oplus \psi_2}{\phi_1 \oplus \phi_2}.
\end{align}
Note that the field operators simply provide a different way to represent $\CAR(\h)$. In particular,
\begin{align}
    a^\dagger(\psi) = B(\psi\oplus 0),\quad a(\phi) = B(\phi\oplus 0)^\dagger.
\end{align}
Given an algebraic state $\omega$ on $\CAR(\h)$, i.e., a positive, normalized linear map $\CAR(\h) \to \CC$, a unique  operator $S$ on $\h \oplus \h$ with $0\le S \le 1$ is defined by
\begin{align}
    \omega(B(\xi)^\dagger B(\eta)) = \ip{\xi}{S \eta}_{\h\oplus \h}. 
\end{align}
Due to the commutation relations of the field operators and since $B(\xi)^\dagger = B(C\xi)$, the operator $S$ fulfills
\begin{equation}
    S = 1 - CSC.
\end{equation}
Indeed,
\begin{align}
    \ip{\xi}{S\eta}_{\h\oplus\h} &= \ip{\xi}{\eta}_{\h\oplus \h} - \omega(B^\dagger(C\eta) B(C\xi)) = \ip{\xi}{\eta}_{\h\oplus \h} - \ip{C\eta}{SC\xi}_{\h\oplus \h} \nonumber\\
    &= \ip{\xi}{\eta}_{\h\oplus \h} - \overline{\ip{CC\eta}{CSC\xi}_{\h\oplus \h} } = \ip{\xi}{\eta}_{\h\oplus \h} - \ip{CSC\xi}{\eta}_{\h\oplus \h} \nonumber\\
    &= \ip{\xi}{\eta}_{\h\oplus \h} - \ip{\xi}{CSC\eta}_{\h\oplus \h}.\nonumber
\end{align}

\begin{definition}
    The operator $S$ is called the \emph{self-dual covariance} of the state $\omega$. 
\end{definition}

\subsection{Majorana formalism}

We now discuss the relation of the self-dual formalism with the Majorana formalism. All notation from the previous section carries over.
Let $V\subset \h\oplus \h$ be the real subspace spanned by vectors $v \in\h\oplus \h$ such that $Cv= v$. 
Since $C^2=1$, we have $\h\oplus \h = V + iV$, i.e., we can decompose any vector $\xi \in \h\oplus\h$ as $\xi = \Re \xi + i \Im \xi$, where $\Re \xi,\Im\xi \in V$, by setting $\Re \xi = \frac{1}{2}(\xi+C\xi)$ and $\Im \xi = -i \frac{1}{2}(\xi - C\xi)$.
Then the field operators associated to the real vectors $v\in V$ in the self-dual formalism fulfill $B(v)^\dagger = B(v)$ and $B(v)^2 = \frac{1}{2}\ip{v}{v}_{\h \oplus \h} =: \ip{v}{v}'_V$, 
where we defined a rescaled inner product $\ip{\cdot}{\cdot}'_V$ on the real Hilbert space $V$.
The commutation relations then take the form
\begin{align}
    \{B(v),B(w)\} = 2\ip{v}{w}'_V. 
\end{align}
These are known as the commutation relations of Majorana operators. 
Since $a^\dagger(\psi) = B(\Re(\psi\oplus 0))+ iB(\Im(\psi\oplus 0))$ the Majorana operators generate the full CAR algebra $\CAR(\h)$.
If we choose $J$ as complex conjugation in an orthonormal basis $\{\phi_i\}_{j=1}^{\dim \h}$ of $\h$, then the vectors
\begin{align}\label{eq:realbasis}
 v_{j} := \phi_j\oplus \phi_j,\quad v_{j+\dim \h} := i\phi_j \oplus (-i)\phi_j,\quad j=1,\ldots,\dim\h,
\end{align}
provide a basis of $V$ that is orthonormal with respect to $\ip{\cdot}{\cdot}'_V$. We thus find
\begin{align}
    \{ B(v_i), B(v_j) \} = 2\delta_{i,j}. 
\end{align}
Consider a self-dual covariance $S$ and define $A:= -i(2S-1)$, so that $S = \frac{1}{2}(1+iA)$. Using the relation $S = 1- CSC$, it follows that $CAC = A$. Hence, $A$ restricts to a real linear operator on $V$ such that for any $v,w\in V$ we have
\begin{align}
    \ip{v}{Aw}_V' &= \frac{1}{2}(v, A w)_{\h\oplus \h} = -\frac{i}{2} \ip{v}{(2S-1)w}_{\h\oplus \h} \\
    &= -i \omega\big(B(v) B(w) - \frac{1}{2}\{B(v),B(w)\}\big) \\
    &= -\frac{i}{2} \omega\big([B(v),B(w)]\big).
\end{align}
In particular, $A$ is anti-symmetric on $V$. 
Note that if $S$ is a projection, we find that $A^2 = -(2 S - 1)^2 = -1$.
\begin{definition}
	The real linear operator $A$ is called the \emph{Majorana covariance} of $\varphi$.
\end{definition}

\subsection{Relation to passive Gaussian states}

Now, consider a state $\omega$ on $\CAR(\h)$ such that $\omega(a^\dagger(\phi)a^\dagger(\psi)) =0$ for all $\psi,\phi\in \h$ (it follows that $\omega(a(\phi)a(\psi))=0$ as well).
This is the case for the passive (gauge-invariant) Gaussian states discussed in the main text.
Let $0\le G\le 1$ be the covariance of $\omega$ on $\h$ determined by $\omega(a(\phi)a^\dagger(\psi))=\ip{\phi}{G\psi}_\h$, $\psi,\phi\in \h$.
Then one finds by a direct calculation that
\begin{align}\label{eq:gauge-invariant}
    S = G\oplus (1-JGJ),\quad A =  \1_2 \ox 2\, \Im G+ i\sigma_z \ox (1 - 2 \,\Re G), 
\end{align}
where the real and imaginary parts of $G$ are defined relative to $J$, i.e., $\Re G = 1/2(G+JGJ)$, $\Im G = -i/2(G-JGJ)$.
Suppose the latter corresponds to complex conjugation relative to an orthonormal basis $\{\phi_j\}$ of $\h$ and consider the matrix representation $\mathbb A_{a,b} :=\ip{v_a}{A v_b}'_V = -(i/2) \omega([B(v_a),B(v_b)])$ of $A$ relative to the induced orthonormal basis $\{v_a\}$ on $V$ (cp. \eqref{eq:realbasis}) with respect to the inner product $\ip{\cdot}{\cdot}'_V$.
Similarly, consider the matrix with entries $\mathbb G_{ij} := \ip{\phi_i}{G\phi_j}$. Then we have 
\begin{align}
\mathbb A&=    
        \1_2 \ox 2\, \Im\mathbb G + i\sigma_y \ox (2\,\Re \mathbb G - \1),
\end{align}
where $\1$ denotes the unit matrix. This is indeed the standard result, see for example \cite{bittel_optimal_2024}.

\subsection{General Gaussian states and entanglement}
We now introduce general Gaussian states. 
\begin{definition}
    An even state $\omega: \CAR(\h)\to \CC$ is called \emph{Gaussian} or \emph{quasi-free} if for any $n\in \NN$ and  $w_1,\ldots,w_n \in V$ we have
    \begin{align}
        \omega\big(B(w_1) \cdots B(w_n)\big) = i^{n/2} \mathrm{Pf}\left(\left[\ip{w_i}{Aw_j}'_V\right]_{i,j=1}^n\right),
    \end{align}
    where $A$ is the Majorana covariance of $\omega$, and $\mathrm{Pf}$ denotes the Pfaffian. 
\end{definition}
Note that the Pfaffian is zero if $n$ is odd.
Of course, one can also directly express the above relation in terms of the self-dual covariance $S$, see for example \cite[Sec. 6.6]{evans1998qsym}. 

A Gaussian state is completely determined by its self-dual covariance $S$ (equivalently, its Majorana covariance $A$).  We hence write $\omega_S$ for the Gaussian state induced by $S$.  
The state is passive whenever $S = G\oplus (1 - JGJ)$, in which case the usual determinant formula stated in the main text holds. 

A Gaussian state is pure if and only if its self-dual covariance is a projection, i.e., $S^2=S$. Such projections, fulfilling $S + CSC = 1$, are also called \emph{basis projections} in the literature. 
Evidently, $S$ is a projection if and only if the Majorana covariance fulfills $A^2 = -1$. 

Consider now a bipartite splitting of the single-particle Hilbert space $\h$ as $\h_A \oplus \h_B$ and denote by $P_{A},P_{B}$ the associated projections. 
We say that the splitting is compatible with the conjugation $J$ if $J P_{A/B} = P_{A/B} J$, i.e., $J = J_A\oplus J_B$. Then $C$ also restricts to $\h_{A/B}\oplus \h_{A/B}\subset \h\oplus \h$, i.e., $\h\oplus \h \cong (\h_A\oplus \h_A) \oplus(\h_B\oplus \h_B)$, relative to which we have $C = C_A\oplus C_B$.

If $\omega_S$ is a pure Gaussian state on $\CAR(\h)$ with self-dual covariance $S$, its restriction to $\CAR(\h_A)$ is a Gaussian state with self-dual covariance $S_A$ (relative to the conjugation $C_A$ on $\h_A\oplus \h_A$) obtained by restricting $S$ to $\h_A\oplus \h_A$.
Conversely, if $\omega_{S_A}$ is a mixed Gaussian state on $\CAR(\h_A)$ with self-dual covariance $S_A$ on $\h_A\oplus \h_A$ relative to conjugation $C_A$, we can obtain a purification $\omega_S$ on $\CAR(\h_A\oplus \h_B)$ with $\h_A\cong \h_B$. 
The precise relationship between $S_A$ and $S$ will be discussed in \cref{lemma:normalform}.
In the following, we only consider the case of bipartite splittings that are compatible with the conjugation $J$. 

\subsection{Gaussian unitary operations}

We now turn to the general set of Gaussian unitary operations, also known as Bogoliubov transformations.
They are induced by unitary operators $u$ on $\h\oplus \h$ such that $CuC = u$, hence inducing an orthogonal transformation on the real subspace $V$. 
Conversely, any orthogonal transformation on $V$ induces a corresponding unitary operator $u$.
Let $P$ be the projection onto the first summand of $\h\oplus\h$ and consider the operators
\begin{align}
    B_u(\xi) := B(u\xi) = a^\dagger(P u \xi) + a(J(1-P)u\xi),\quad \xi\in \h\oplus\h. 
\end{align}
Then, since $CuC=u$, we have $B_u(C\xi) = B(uC\xi) = B(Cu\xi) = B(u\xi)^\dagger = B_u(\xi)^\dagger$ and
\begin{align}
    B_u(\xi)B_u(\eta)^\dagger + B_u(\eta)^\dagger B_u(\xi) = \ip{u\eta}{u\xi} = \ip{\eta}{\xi}.
\end{align}
Thus, the operators $B_u$ provide again self-dual fields relative to the conjugation $C$.
Hence, we can define an automorphism $\alpha_u$ of $\CAR(\h)$ by $B(\xi)\mapsto B_u(\xi)$.  

Suppose that $\omega_S$ is pure, so that we can think of it as a vector $\Omega_S$ in a Fock \rm{space} representation of $\CAR(\h)$. 
Then the automorphism $\alpha_u$ can be represented by a unitary $U$ on the Fock space if and only if $\norm{[S,u]}_2<\infty$ \cite{evans1998qsym}. This is, of course, always the case if $\h$ is finite-dimensional. 
\begin{definition}
    The automorphisms $\alpha_u$ are called Gaussian unitary operations.
\end{definition}

If $\omega$ is a Gaussian state with self-dual covariance $S$, then the state $\omega\circ \alpha_u$ is also Gaussian with self-dual covariance $S_u := u^\dagger S u$. 
Indeed,
\begin{align}
    \ip{\xi}{S_u\eta} = \omega(B_u(\xi)^\dagger B_u(\eta)) = \omega(B(u\xi)^\dagger B(u\eta)) = \ip{u\xi}{S u\eta} = \ip{\xi}{u^\dagger S u\eta}.
\end{align}
If the dimension of the $+\frac{1}{2}$ eigenspace of $S$ is either even or infinite, which is always fulfilled if $\h$ is finite-dimensional,  it follows from the relation $S + CSC=1$ that we can always find a unitary $u$ with $CuC=u$ such that 
\begin{align}
    S_u = G \oplus (1-JGJ)
\end{align}
for some $0\leq G\leq 1$ on $\h$.\footnote{A sketch is as follows: From the relation $S=1-CSC$ and the assumption on the $+1/2$ eigenspace of $S$ it follows that we can decompose $\k=\h\oplus\h$ as $\k=\k_+\oplus \k_-$ and (relative to this decomposition) $S=S_+\oplus S_-$  in such a way that $S_+\geq 1/2, S_-\leq 1/2$ and $C \k_+ = \k_-$.
It follows that $\h$ and $\k_+$ are isomorphic. Choose any unitary $v:\h\to \k_+$ and set $u = v\oplus 0 + C(v\oplus 0)C$. Then $u^\dagger S u$ takes the appropriate form.}
In other words, there is always a Gaussian unitary operation that transforms a Gaussian state into a passive Gaussian state. 
We emphasize that the resulting covariance $G$ is highly non-unique. This is illustrated by the following example.

\begin{example}\label{example:bogoliubov}
While passive Gaussian operations conserve the total number of fermions, general Gaussian unitary operations do not. Indeed, the unitary $u$ acting as $\xi_1\oplus \xi_2 \mapsto \xi_2 \oplus \xi_1$ is a valid Bogoliubov transformation, but we have
\begin{align}
 \alpha_u(a^\dagger(\psi) )  = B_u(\psi \oplus 0) = a(J\psi), 
\end{align}
so that $u$ effectively transforms creation operators into annihilation operators. If $\phi\in \h$ is normalized and real relative to $J$ we thus have that $n(\phi) = a^\dagger(\phi)a(\phi) \mapsto a(\phi)a^\dagger(\phi) = 1 - n(\phi)$. A passive gauge-invariant state with covariance $G$ is thus mapped to one with covariance $1-J GJ$. The same transformation may also be applied to $J$-invariant subspaces. This allows us to transform a projection into any other projection, thereby mapping between Fock states (pure states with fixed particle numbers) using Bogoliubov transformations. This does not violate the parity super-selection rule.
\end{example}

The passive Gaussian unitary operations, which are studied in the main text and are induced by a unitary operator $v$ on $\h$. In the self-dual formalism, they are represented on $\h\oplus \h$ by $u = v\oplus JvJ$. 
Any passive Gaussian unitary operation maps passive Gaussian states to passive Gaussian states. The example above shows, however, that there are non-passive Gaussian unitary operations that also map all passive Gaussian states to passive Gaussian states.

Now let $\omega_S$ be a Gaussian state on $\CAR(\h)$ and let $\h=\h_A\oplus \h_B$ be a bipartition compatible with the conjugation $J$ (i.e., $J = J_A\oplus J_B$ for local conjugations $J_{A/B}$). 
Then, we can find unitaries $u_A$ on $\h_A\oplus \h_A$ and $u_B$ on $\h_B\oplus \h_B$ with $C_{A/B}u_{A/B}C_{A/B} = u_{A/B}$ such that the Bogoliubov transformation corresponding to $u_A\oplus u_B$ takes the state $\omega_S$ to a state with passive marginals.

\begin{definition}
    Let $\omega$ be a state on $\CAR(\h_A\oplus \h_B)$ and $\varphi,\psi$ be states on $\CAR(\k_A\oplus \k_B)$. We say that the state-transition $\varphi\rightarrow \psi$ can be embezzled from $\omega$ up to error $\eps>0$ via \emph{local Gaussian operations} if
    \begin{align}
        \frac{1}{2}\norm{(\varphi \ox \omega)\circ \alpha_{u_A\oplus u_B} - \psi \ox \omega} \le \eps,
    \end{align}
    for unitaries $u_A$ on $(\k_A\oplus \h_A)\oplus (\k_A\oplus \h_A)$ and $u_B$ on  $(\k_B\oplus \h_B)\oplus (\k_B\oplus \h_B)$ such that $C_{A/B}u_{A/B}C_{A/B} = u_{A/B}$.
\end{definition}

\subsection{The normal form of the self-dual covariance of a bipartite pure state}
The following Lemma provides a normal form for self-dual covariances. 
\begin{lemma}\label{lemma:normalform}
    Let $\h = \h_A\oplus \h_B$ and $S$ be a basis-projection on $\h \oplus \h \cong (\h_A \oplus \h_A) \oplus (\h_B \oplus \h_B)$ relative to
    \begin{align}
        C = \begin{pmatrix} 0 & J_A\oplus J_B \\ J_A\oplus J_B & 0 \end{pmatrix} \cong  C_A \oplus C_B,
    \end{align} 
    where $C_{A/B} = \begin{pmatrix} 0 & J_{A/B} \\ J_{A/B} & 0 \end{pmatrix}$. 
    Then there exist decompositions $\h_A\oplus \h_A = \n_A \oplus \r_A, C_A = C_A|_{\n_A} \oplus C_A|_{\r_A}$ (and similarly for $B$), and a unitary $\tilde v:\r_A \to \r_B$ fulfilling $\tilde v C_A|_{\r_A}=-C_B|_{\r_B}\tilde v$ such that 
    \begin{align}
        S = Q_A \oplus \begin{pmatrix}
            \tilde S_A & \tilde S_A^{1/2} (1-\tilde S_A)^{1/2} \tilde v^\dagger\\
            \tilde v \tilde S_A^{1/2} (1-\tilde S_A)^{1/2}& \tilde v(1-\tilde S_A)\tilde v^\dagger 
        \end{pmatrix} \oplus Q_B,\label{eq:self-dual-purification}
    \end{align}
    where $Q_{A/B}$ are basis projections on $\n_{A/B}$ and $0<\tilde S_A<1$ is a self-dual covariance on $\r_A$.\footnote{Here, $0<X <1$ for an operator $X$, means that $0\le X\le 1$ and that $\ker(X)=\ker(1-X) = \{0\}$ (in infinite dimension this is not the same as $0\notin \spec(X)$).}
\end{lemma}
\begin{proof}
    Write $S = \begin{pmatrix} S_A & S_{AB} \\ S_{BA} & S_B\end{pmatrix}$ with $S_{AB} = S_{BA}^\dagger$. 
     We now set $\r_A = \supp(S_A(1-S_A))$ and $\n_A = \ker(S_A(1-S_A))$
    as well as $\tilde S_A = S_A|_{\r_A}$ and $Q_A = \tilde S_A|_{\n_A}$. Clearly $\spec(Q_A)\subset\{0,1\}$, $\tilde S_A(1-\tilde S_A) = (S_A(1-S_A))|_{\r_A}$. Similar relations hold for $B$.
    
    From $S^2=S$ we find $|S_{BA}|^2 = S_{BA}^\dagger S_{BA}  = S_A(1-S_A)$ and $S_A S_{AB} + S_{AB}S_B=S_{AB}$ and similarly with $A,B$ exchanged. From the first relation and the polar decomposition there exists a partial isometry $v:\h_A\oplus \h_A \to \h_B\oplus \h_B$ such that $S_{BA} = v |S_{BA}| = v S_A^{1/2} (1-S_A)^{1/2}$ and such that $v^\dagger v$ is the projection onto $\r_A$. Since $S_{AB}=S_{BA}^\dagger$, we also find $v S_A(1-S_A) = S_B(1-S_B)v$ and $v v^\dagger$ is the projection onto $\r_B$. In particular $S_{AB} = v^\dagger S_B^{1/2} (1-S_B)^{1/2}$ restricts to an invertible map $\tilde S_{AB}:\r_B \to \r_A$ and vanishes on $\n_B$. 
     
    From the second relation above, we find $v S_{AB} S_B = v (1-S_A) S_{AB}$.
    Hence $vv^\dagger S_B^{1/2} (1-S_B)^{1/2} S_B vv^\dagger  = v(1-S_A)v^\dagger S_B^{1/2}(1-S_B)^{1/2}$. Dividing by $S_B^{1/2}(1-S_B)^{1/2}$ on its support $\r_B$ hence yields
    \begin{align}
      vv^\dagger S_B vv^\dagger  = v(1-S_A)v^\dagger. 
    \end{align}
    The partial isometry $v$ thus restricts to a unitary $v:\r_A\to \r_B$ such that $\tilde S_B = \tilde v(1-\tilde S_A)\tilde v^\dagger$ and $\tilde S_{AB} = \tilde S_A^{1/2}(1-\tilde S_A)^{1/2} \tilde v^\dagger$. 
        
    Finally, from $S + CSC = 1$ we find $C_A S_{AB} C_B = - S_{AB}$, which implies $C_A v C_B = -v$, and $C_A S_A C_A = (1-S_A)$ (on $\h_A\oplus \h_A$). The remaining claims then follow immediately.
\end{proof}
The part of $S_B$ on $\n_B$ corresponds to a pure state. The Lemma shows that the entanglement is completely determined by $\tilde S_A$, corresponding to the faithful part of the reduced state on $\CAR(\h_A)$. 
Moreover, by a Gaussian unitary operation, we can bring $\tilde S_A$ into the form $G_A\oplus (1-J_AG_AJ_A)$ with $G_A$ invertible, showing that the entanglement is completely equivalent to that of a passive, faithful Gaussian state with correlation operator $G_A$. 

\begin{corollary}\label{cor:normalform}
    Let $\omega_S$ be a pure Gaussian state on $\CAR(\h_A\oplus \h_B)$ with faithful marginals on $\CAR(\h_{A/B})$. Then there exists  a unitary $u :\h_A\oplus \h_A\to \h_B\oplus \h_B$ such that
    \begin{align}
        S = \begin{pmatrix}
            S_A & S_A^{1/2}u^\dagger S_B^{1/2}  \\
            S_B^{1/2}u S_A^{1/2} & S_B,
        \end{pmatrix} = (1 \oplus u) \begin{pmatrix}
            S_A & S_A^{1/2}(1-S_A)^{1/2} \\
            S_A^{1/2}(1-S_A)^{1/2} & 1 - S_A
        \end{pmatrix}(1 \oplus u^\dagger),
    \end{align} 
    and $S_B = u(1-S_A)u^\dagger$. 
\end{corollary}
\begin{remark}\label{remark:normalform-passive}
    If $\omega_S$ is a passive Gaussian state, so that $S = G \oplus J(1-G)J$ (in the ordering $(\h_A\oplus \h_B) \oplus (\h_A\oplus \h_B)$), it follows that there exists a unitary $u: \h_A\to \h_B$ such that
    \begin{align}
        S = \begin{pmatrix}
            G_A& G_A^{1/2}(1-G_A)^{1/2} u^\dagger & &\\
            u G_A^{1/2}(1-G_A)^{1/2}  & u(1-G_A) u^\dagger&& \\
            && J_A(1-G_A)J_A  &-J_A G_A^{1/2}(1-G_A)^{1/2} u^\dagger J_B \\
            && - J_B u G_A^{1/2}(1-G_A)^{1/2} & J_B u G_A u^\dagger J_B
        \end{pmatrix}.
    \end{align}
    In particular $G_B = u(1-G_A)u^\dagger$. Note that we can write $S = \big[(1\oplus u)P_{G_A}(1\oplus u)^\dagger \big]\oplus \big[(1\oplus J_B u J_A)(1-P_{J_A G_A J_A})(1\oplus J_B u J_A)^\dagger\big]$. 
    Since passive Bogoliubov transformations are of the form $ \tilde u\oplus J\tilde u J$ with $\tilde u$ a unitary on $\h$, this general form is of course preserved under passive Gaussian operations. 
\end{remark}

\subsection{Distance bounds}\label{sec:eta}
For positive contractions $S,T$ on a Hilbert space, define
\begin{align}\label{eq:eta}
	\eta(S,T) := \norm{S^{1/2} (\1-T)^{1/2} - (\1-S)^{1/2} T^{1/2}}_2. 
\end{align}
If $S,T$ are projections, we have $\eta(S,T) = \norm{S-T}_2$. 
If we write $S = \frac{1}{2}(1 + iA),$ $T = \frac{1}{2}(1 + iB)$ in terms of Majorana covariances $A,B$, then 
	\begin{align}
	\eta(S,T) = \frac{1}{2}\norm{(1+iA)^{1/2}(1-iB)^{1/2} - (1 - iA)^{1/2}(1+iB)^{1/2}}_2.
	\end{align}
A passive Gaussian state is completely determined by the covariance $G$. If $S,T$ are self-dual covariances corresponding to passive Gaussian states with covariances $F,G$ we have, using \eqref{eq:gauge-invariant}, 
\begin{align}
	\eta(S,T) = \sqrt{2}\, \eta(F,G). 
\end{align}

\begin{proposition}\label{app:prop:cov-bound}
    Let $\omega_S,\omega_T$ be Gaussian states on $\CAR(\h)$ such that $\norm{S^{1/2}-T^{1/2}}_2<\infty$. Then
    	\begin{align}
		1 - \exp(-\frac{1}{4}\eta(S,T)^2) \leq 	\frac{1}{2}\norm{\omega_S - \omega_T} \leq  \eta(S,T). 
	\end{align}
	In particular, if $\omega_S,\omega_T$ are passive (gauge-invariant) with covariances $F,G$, then
	\begin{align}
		1 - \exp(-\frac{1}{2}\eta(F,G)^2)	\leq 	\frac{1}{2}\norm{\omega_S - \omega_T} \leq \sqrt{2}\, \eta(F,G). 
	\end{align}
\end{proposition}

We have here expressed the result in terms of the norm-distance between functionals, which translates into the trace-distance bound stated in the main text when we represent the states as density matrices on Fock space.

\begin{lemma}\label{lemma:cov-bound-trace-distance}
    Let $0\leq S,T\leq 1$. Then 
    \begin{align}
    \eta(S,T)^2 \leq 2\norm{S-T}_1.
    \end{align}    
    Hence $\norm{\omega_S-\omega_T} \leq 2^{3/2}\sqrt{\norm{S-T}_1}$.
\end{lemma}
\begin{proof}[Proof of \cref{lemma:cov-bound-trace-distance}]
    We have
    \begin{align}
        \eta(S,T)^2&=  \norm{\sqrt{1-S}\sqrt{T}-\sqrt{S}\sqrt{1-T}}_2^2 = \norm{\sqrt{1-S}(\sqrt{T}-\sqrt{S}) - \sqrt{S}(\sqrt{1-T}-\sqrt{1-S})}_2^2\\
	&\leq \norm{1-S}\norm{\sqrt{T}-\sqrt{S}}_2^2 + \norm{S}\norm{\sqrt{1-T}-\sqrt{1-S}}_2^2 \\
	&\leq 2 \norm{S-T}_1,
    \end{align}
    where we used the triangle-inequality, $\norm{AB}^2_2 \leq \norm{A^\dagger A} \norm{B}_2^2$ and $\norm{S^{1/2} - T^{1/2}}_2^2 \leq \norm{S-T}_1$  \cite{powers1970free_states}.
\end{proof}
If $0\leq S\leq 1$ is an operator on a Hilbert space $\k$, we consider the following projection on $\k \oplus \k$
\begin{align}\label{eq:cov_proj}
    P_S := \begin{pmatrix}
        S & S^{1/2}(1-S)^{1/2}\\
        S^{1/2} (1-S)^{1/2} & 1-S
    \end{pmatrix} = v_S v_S^\dagger,
\end{align}
where $v_S: \k\to \k\oplus \k$ is the isometry acting as $\psi \mapsto S^{1/2}\psi \oplus (1-S)^{1/2} \psi$. If $w_S:\k\to \k\oplus \k$ is the isometry acting as $\psi \mapsto (-(1-S)^{1/2})\psi\oplus S^{1/2}\psi$, we have $w_S w_S^\dagger = 1 - P_S$ and 
\begin{align}
    \eta(S,T)^2 = \norm{v_S^\dagger w_T}_2^2. 
\end{align}
\begin{lemma}\label{lemma:eta-purification}
    Let $0\leq S,T\leq 1$. Then 
    \begin{align}
        \eta(P_S,P_T)^2 = 2\, \eta(S,T)^2. 
    \end{align}
    \begin{proof}
    Using the cyclicity of the trace, we have
        \begin{align}
            \eta(S,T)^2 = \norm{v_S^\dagger w_T}_2^2 = \tr(v_S v_S^\dagger w_T w_T^\dagger) = \tr(P_S(1-P_T)) = \tr(P_T(1-P_S)),
        \end{align}
        where the last equality follows from the symmetry of the left-hand side under $S\leftrightarrow T$.
    Since $P_S(1-P_T) + P_T(1-P_S) = (P_S-P_T)^2$, the claim follows.
    \end{proof}
\end{lemma}

\begin{remark}
    In \cite{bittel_optimal_2024} the bound
    \begin{align}
        \norm{\omega_S - \omega_T} \leq \frac{1}{2}\norm{\mathbb A - \mathbb B}_1 = \frac{1}{2}\norm{A-B}_1 = \norm{S-T}_1
    \end{align}
    was derived in terms of the trace-norm of the (matrix representations of the) respective Majorana covariances $A$ and $B$.  If we have passive Gaussian states, so that $S = G\oplus (1-JGJ)$ and $T=F\oplus (1-JFJ)$, we  find
    \begin{align}
        \norm{\omega_S-\omega_T} \leq 2 \norm{G-F}_1,
    \end{align}
    i.e., \eqref{eq:bittel} in the main text.
    The derived bound in \cref{lemma:cov-bound-trace-distance} is  weaker than the above bound, but our application to embezzlement shows that the bound in \cref{app:prop:cov-bound} can be significantly tighter. 
\end{remark}

\begin{proof}[Proof of \cref{app:prop:cov-bound}.] To set up the proof, we briefly recall some facts from the theory of von Neumann algebras.
To any quasi-free state $\omega_S$ on $\CAR(\h)$ we can associate a von Neumann algebra $\M_S := \pi_S(\CAR(\h))''$ via the GNS representation $\pi_S$ of $\omega_S$. 
If $\norm{\sqrt S- \sqrt T}_2 < \infty$, the states $\omega_S,\omega_T$ are \emph{quasi-equivalent} \cite{araki1970car_bog}.
This implies that we can find a representation of $\M_S$ on a Hilbert space $\H_S$ (its standard representation), such that there are normalized vectors $\Omega_S,\Omega_T \in \H_S$ in the associated \emph{positive cone} of $\H_S$ which realize the two states and fulfill $\braket{\Omega_S}{\Omega_T}\geq 0$ \cite{haagerup_standard_1975}.
Then the \emph{transition-probability} $P(\omega_S,\omega_T)$ is defined as $P(\omega_S,\omega_T) := \braket{\Omega_S}{\Omega_T}^2$. 
For two density operators $\rho,\sigma$ on a finite-dimensional Hilbert space, the vectors associated vectors $\Omega_\rho,\Omega_\sigma$ simply correspond to canonical purifications and we have $P(\rho,\sigma) = \tr(\rho^{1/2}\sigma^{1/2})^2$.
Note that the transition probability is distinct from Uhlmann's fidelity $\mathrm{F}(\rho,\sigma) = \norm{\rho^{1/2}\sigma^{1/2}}_1$.

The Powers-St\o rmer inequality \cite{powers1970free_states} implies
\begin{align}\label{eq:PS}
    1 - \sqrt{P(\omega_S,\omega_T)} \leq \frac{1}{2}\norm{\omega_S-\omega_T} \leq \sqrt{1-P(\omega_S,\omega_T)}.
\end{align}
Araki showed \cite{araki1970car_bog} (see also \cite{matsuiKakutaniDichotomyFree2012}) that the transition probability takes a particularly simple form for Gaussian states:\footnote{In infinite dimensions, if $0\leq A \leq 1$, the determinant is defined via $\det(A):=\inf_V\det (A|_V)$ where the infimum is taken over all finite-dimensional subspaces, see \cite{powers1970free_states} for details. 
}
\begin{align}\label{eq:araki}
    P(\omega_S,\omega_T) = \det(MM^\dagger)^{1/2},
\end{align}
where $M = v_S^\dagger v_T = S^{1/2}T^{1/2} + (1-S)^{1/2} (1-T)^{1/2}$. 
In particular, we have
\begin{align}
    0\leq MM^\dagger = v_S^\dagger v_T v_T^\dagger v_S = v_S^\dagger P_T v_S \leq 1
\end{align}
and hence $\det(MM^\dagger)\leq 1$.
Using Araki's formula \eqref{eq:araki} together with the Powers-St\o rmer inequality \eqref{eq:PS} yields
\begin{align}
      1 - \det(MM^\dagger)^{1/4} \leq \frac{1}{2}\norm{\omega_S-\omega_T} \leq \sqrt{1-\det(MM^\dagger)^{1/2}}.
\end{align}

{\bf Upper bound:} For any finite collection of numbers $\lambda_i \in [0,1]$ it is true that\footnote{Proof by induction on the number of elements $n$: For $n=1$ it is trivial. 
Assume it is true for $n$, then
$\prod_{i=1}^{n+1}(1-\lambda_i)=(\prod_{i=1}^n (1-\lambda_i))(1-\lambda_{n+1}) \geq(1-\sum_{i=1}^n \lambda_i)(1-\lambda_{n+1}) = 1 + \lambda_{n+1}\sum_{i=1}^n\lambda_i - \sum_{i=1}^{n+1} \lambda_i \geq 1 - \sum_{i=1}^{n+1}\lambda_i$.
}
\begin{align}
	\prod_i (1- \lambda_i) \geq 1-\sum_i \lambda_i.
\end{align}

As a consequence we have $\det(1-X) \geq 1- \tr(X)$ for any $0\leq X \leq 1$.
Thus $\det(X) = \det(1-(1-X)) \geq 1 - \tr(1-X)$ and hence $1 - \det(X) \leq \tr(1-X)$.  
Consequently
\begin{align}
	1 - \det(MM^\dagger)^{1/2} &\leq 1-\det(MM^\dagger) \leq \tr(1-MM^\dagger)  \\
	&= \tr(v_S^\dagger(1-P_T)v_S)\\
	&= \tr(v_S^\dagger w_T w_T^\dagger v_S)\\
	&= \norm{v_S^\dagger w_T}_2^2
    =\eta(S,T)^2.
\end{align}

{\bf Lower bound: } Using $\det(X) = \exp(\tr(\log(X)))$ and $\log(1+x)\leq x$ we find
\begin{align}
	1 - \det(MM^\dagger)^{1/4} &= 1 - \exp\left(\frac{1}{4}\tr\Big(\log\big(1-(1-MM^\dagger)\big)\Big)\right) \geq 1 - \exp\big(-\frac{1}{4}\tr(1-MM^\dagger)\big)\\
    &= 1 - \exp(-\frac{1}{4}\eta(S,T)^2),
\end{align}
using again that $\Tr(1-MM^\dagger) = \eta(S,T)^2$.
\end{proof}

Finally, let us note that there is an alternative expression for $\eta(S,T)$ in terms the Hilbert-Schmidt distance of unitaries $U_{S},U_{T}$ on $\k\oplus\k\cong\k\ox\CC^{2}$ satisfying 
\begin{align}
\label{eq:cov_unitary}
P_{S} & = U_{S}(1\ox\tfrac{1}{2}(\1_{2}+\sigma_{z}))U_{S}^{\dagger},
\end{align}
and analogously for $T$ (cp.~\eqref{eq:cov_proj}). Explicitly, they take the form
\begin{align}\label{eq:cov_unitary_form}
    U_{S} & := \begin{pmatrix} v_S& w_S\end{pmatrix} =  \begin{pmatrix}
        S^{1/2} & -(1-S)^{1/2}\\
        (1-S)^{1/2} & S^{1/2}
    \end{pmatrix} = S^{1/2}\ox\1_{2}-i(1-S)^{1/2}\ox\sigma_{y},
\end{align}
with an analogous expression for $U_{T}$.
\begin{lemma}\label{lem:eta_unitary}
Let $0\leq S,T\leq 1$. Then
\begin{align}\label{eq:eta_unitary}
    \eta(S,T) & = \tfrac{1}{2}\norm{U_{S}^{2}-U_{T}^{2}}_{2}.
\end{align}
\end{lemma}
\begin{proof}
 Using the definition of the operator $M$ (given below \eqref{eq:araki}), \eqref{eq:cov_unitary_form} implies
\begin{align}
    \Tr(1+U_{S}U_{1-T}^{2}U_{S}) & = 2\Tr(1-M^{\dagger}M).
\end{align}
Now, we consider the complex structure $I$ on $\k\oplus\k$, defined by $I(\xi_{1}\oplus\xi_{2})=(-\xi_{2})\oplus\xi_{1}$, and infer from \eqref{eq:cov_unitary_form} the relations $U_{1-S}I = -U_{S}^{\dagger}$ and $IU_{1-S} = -U_{S}^{\dagger}$. This leads to
\begin{align} \nonumber
    \Tr(1-M^{\dagger}M) & = \tfrac{1}{2}\Tr(1+U_{S}U_{1-T}^{2}U_{S}) = \tfrac{1}{2}\Tr(1+U_{S}U_{1-T}I(-I)U_{1-T}U_{S}) \\ & = \tfrac{1}{2}\Tr(1-U_{S}U_{T}^{\dagger}U_{T}^{\dagger}U_{S}) = \tfrac{1}{2}\Tr(1-(U_{T}^{\dagger})^{2}(U_{S})^{2}).
\end{align}
By symmetry in $S,T$, we obtain
\begin{align} \nonumber
    \Tr(1-M^{\dagger}M) & = \tfrac{1}{4}\Tr(2-(U_{T}^{\dagger})^{2}(U_{S})^{2}-(U_{S}^{\dagger})^{2}(U_{T})^{2}) = \tfrac{1}{4}\norm{U_{S}^{2}-U_{T}^{2}}_{2}^{2},
\end{align}
which completes the proof.
\end{proof}

\subsection{Reduction to the monopartite case}
We will now show that the distance of two local self-dual covariances $S_A,T_A$ can be used to bound the distance of purifications up to local Bogoliubov transformations. 
This reduces the problem of Gaussian embezzlement of entanglement to a problem on simple (\emph{monopartite}) quantum systems described by mixed Gaussian states. 

\begin{lemma}\label{lemma:mono-to-bipartite-quasifree}
   Let $S,T$ be basis projections on $(\h_A\oplus \h_B)\oplus(\h_A\oplus \h_B) \cong (\h_A\oplus \h_A)\oplus(\h_B\oplus \h_B)$ that induce faithful Gaussian states on $\CAR(\h_A)$ and $\CAR(\h_B)$, respectively.   If there exists a Bogoliubov transformation $\CAR(\h_A)$ induced by a unitary $u_A$ on $\h_A\oplus \h_A$ such that
   \begin{align}
       \eta(S_A, u_A T_A u_A^\dagger) \leq \eps,
   \end{align}
   then there exists a Bogoliubov transformation on $\CAR(\h_B)$ such that
   \begin{align}
       \eta\big(S, (u_A\oplus u_B) T (u_A \oplus u_B)^\dagger\big) \leq \sqrt{2}\, \eps.
   \end{align}
\end{lemma}
\begin{proof}
    By \cref{cor:normalform} we can write $S = (1\oplus \tilde u_S) P_{S_A} (1\oplus \tilde u_S)^\dagger$, $T = (1\oplus \tilde u_T) P_{T_A} (1\oplus \tilde u_T)^\dagger$, where $\tilde u_S,\tilde u_T:\h_A\oplus \h_A \to \h_B\oplus \h_B$ are unitary. 
    Define $u_B = \tilde u_S u_A \tilde u_T^\dagger$. Then 
    \begin{align}
        (u_A\oplus u_B)T(u_A\oplus u_B)^\dagger = (1\oplus \tilde u_S)   (u_A\oplus u_A) P_{T_A}   (u_A\oplus u_A)^\dagger (1\oplus \tilde u_S)^\dagger = (1\oplus \tilde u_S) P_{u_A T_A u_A^\dagger} (1\oplus \tilde u_S)^\dagger.
    \end{align}
    Hence by \cref{lemma:eta-purification}  and unitary invariance of $\eta$ we have
    \begin{align}
        \eta\big(S, (u_A\oplus u_B) T (u_A \oplus u_B)^\dagger\big) = \eta\big(P_{S_A}, P_{u_A T_A u_A^\dagger}\big) = \sqrt{2} \eta(S_A, u_A T_A u_A^\dagger) \leq \sqrt{2}\,\eps.
    \end{align}
\end{proof}

\begin{corollary}
    Let $\omega$ be a pure Gaussian state on $\CAR(\k_A\oplus \k_B)$ and $\psi,\varphi$ be pure Gaussian states on $\CAR(\h_A\oplus \h_B)$, respectively. Suppose there exists a Gaussian unitary operation $\alpha_{u_A}$ on $\CAR(\h_A\oplus \k_A)$ such that
    \begin{align}
        \frac{1}{2}\norm{(\varphi_A\ox \omega_A)\circ \alpha_{u_A} - \psi_A\ox \omega_A } \leq \eps \leq \frac{1}{2}, 
    \end{align}
    where $\varphi_A$ denotes the marginal on $\CAR(\h_A)$ (similarly for $\psi,\omega$). 
    Then there exists a Gaussian unitary operation $\alpha_{u_B}$ on $\CAR(\h_B\oplus \k_B)$ such that
    \begin{align}
         \frac{1}{2}\norm{(\varphi\ox \omega)\circ \alpha_{u_A}\circ \alpha_{u_B} - \psi \ox \omega } \leq 4 \eps^{1/2},
    \end{align}
    where we we identify $u_A$ with $u_A\oplus 1$ (and similarly for $u_B$). 
\end{corollary} 
\begin{proof}
    Denote by $R,S,T$ the self-dual covariances of $\omega,\psi,\varphi$, respectively.
    From the lower bound in \cref{prop:cov-bound} we find that 
    $\eta( R_A\oplus S_A, u_A(R_A\oplus T_A) u_A^\dagger)^2 \leq -4 \log(1-\eps)$. By \cref{lemma:mono-to-bipartite-quasifree} there exists a Bogoliubov transformation $u_B$
    such that
    $\eta(R\oplus S, (u_A\oplus u_B) R\oplus T(u_A\oplus u_B))  \leq \sqrt{2}(- \, 4 \log(1-\eps))^{1/2} \leq 4 \log(2)^{1/2} \eps^{1/2} \leq 4\eps^{1/2}$, where we used that $-\log(1-\eps)\leq 2\log(2)\eps$ for $0\leq \eps\leq 1/2$.  The result now follows from the upper bound in \cref{prop:cov-bound}.
\end{proof}
\begin{remark}
    In the case of passive Gaussian states, a similar statement of course holds for passive Gaussian operations using the normal form discussed in \cref{remark:normalform-passive}. 
\end{remark}

\subsection{Reduction to the passive case}
In the previous subsection, we have shown that the problem of embezzlement of entanglement can be reduced to the monopartite case. We now argue that the general case can be reduced to the passive case. 
Indeed, using the unitary freedom in the definition of embezzlement, we can directly deduce that the monopartite embezzlement problem only depends on the spectra of the self-dual covariances involved. These are in one-to-one correspondence with those of passive Gaussian states, since we can turn every Gaussian state into a passive one using a Bogoliubov transformation (which may be undone at the end of the protocol).
We have thereby reduced the general problem of embezzlement to that of embezzling monopartite, passive Gaussian states as treated in the main text. 

\section{Inadequacy of the trace-distance of covariances}

In this section, we show the inequality
\begin{align}
    \inf_{u}\, \norm{u(F\oplus K)u^\dagger - G\oplus K}_1 \geq \inf_{u}\, \norm{uFu^\dagger- G}_1,
\end{align}
for arbitrary (finite-dimensional) Hermitian matrices $F,G,K$, and where the optimization runs over all unitary matrices $u$. 
We make use of the fact that, for every pair of Hermitian $n\times n$ matrices $A,B$, we have
\begin{align}\label{eq:distance-unitary-orbits}
    \inf_{u}\,\norm{uAu^\dagger - B}_1 = \norm{\lambda(A)^\downarrow -\lambda(B)^\downarrow}_1,
\end{align}
\Cref{eq:distance-unitary-orbits} follows from Wielandt's theorem and holds for any Schatten $p$-norm (see also \cite{hiai1989distance} for a generalization to von Neumann algebras).
To simplify notation, we will write $f=\lambda(F)^\downarrow$, $g=\lambda(G)^\downarrow \in \RR^d$, $k=\lambda(K)^\downarrow\in \RR^m$.
By the above result, we need to show that
\begin{align}
  \norm{(g\oplus k)^\downarrow - (f\oplus k)^\downarrow}_1\geq   \norm{g - f}_1.
\end{align}
Our prove will work by induction on the dimension $m$ of the vector $k$. 
We begin with $m = 1$, i.e. $k=(k_1)$.
For a subset $X\in\NN$, let $g_{X}$ denote the sub-vector containing the entries $g_j$ with $j\in X$. 
Then $(g\oplus k)^\downarrow = g_{[1,j]} \oplus k \oplus g_{(j,d]}$ and $(f\oplus k)^\downarrow = f_{[1,l]} \oplus k\oplus f_{(l,d]}$ for some $j,l\in \NN$.
Without loss of generality assume $l\geq j$, so that $k_1\geq g_i$ for $i\geq j+1$ and $f_i \geq k_1$ for $i\leq m$.
Then 
\begin{align}
    \norm{(g\oplus k)^\downarrow - (f\oplus k)^\downarrow}_1 &= \norm{g_{[1,j]}-  f_{[1,j]}}_1 + \norm{k\oplus g_{[j+1,m]}- f_{[j+1,m]} \oplus k}_1 + \norm{g_{[m+1,d]}-f_{[m+1,d]}}_1.
\end{align}
We have
\begin{align}
    \norm{k\oplus g_{[j+1,m]} - f_{[j+1,m]} \oplus k}_1 &= \sum_i \big|(k\oplus g_{[j+1,m]})_i - (f_{[j+1,m]} \oplus k)_i \big| \\
    &= (f_{j+1} - k_1) + \sum_{i=j+2}^m(f_i - g_{i-1}) + (k_1 - g_{m}) \\
    &= \sum_{i=j+1}^m(f_i - g_i) = \norm{g_{[j+1,m]} - f_{[j+1,m]}}_1.
\end{align}
Thus $\norm{(g\oplus k)^\downarrow - (f\oplus k)^\downarrow}_1 = \norm{g-f}_1$.
Now, suppose that we have shown the statement for all dimensions of $k$ up to $m$. If $k\in \RR^{m+1}$ we write $k=k_{[1,m]} \oplus k_{\{m+1\}}$. 
Using the short-hands $\tilde g=(g\oplus k_{[1,m]})^\downarrow$ and $\tilde f= (f\oplus k_{[1,m]})^\downarrow$, we find
\begin{align}
    \norm{(g\oplus k)^\downarrow - (f\oplus k)^\downarrow}_1 &= \norm{(\tilde g \oplus k_{\{m+1\}})^\downarrow - (\tilde f\oplus k_{\{m+1\}})^\downarrow}_1 \\
    &= \norm{\tilde g - \tilde f}_1 = \norm{g-f}_1,
\end{align}
finishing the proof.

\section{Entropy and density of spectrum for Gaussian states}

We are interested in the relation between the spectral density of the covariance matrix and the von Neumann entropy of a Gaussian state.
We can assume without loss of generality that the state in question is passive; otherwise, there will be an (active) Gaussian unitary that makes it passive and does not affect the entropy.

We consider a fermionic system with finitely many modes.  Let $\rho_K$ be a passive Gaussian state with covariance $K$.
If $\lambda_i(K)$ denote the eigenvalues of $\rho_K$, a simple calculation shows 
\begin{align}
    H(\rho_K) = \sum_j h(\lambda_j(K)),
\end{align}
where $h(x) = -x\log(x) - (1-x)\log(1-x)$ is the binary entropy. 
If $K$ has $\eps$-dense spectrum, for any $0<\delta < 1/2$, the number of spectral values in the range $[\delta,1-\delta]$ is lower-bounded by $n_{\eps,\delta} = \lfloor(1-2\delta)/\eps\rfloor$. Since the binary entropy is concave, we find that $h(x) \geq 2 x \log(2)$ for any $x\in [0,1/2]$. We thus find
\begin{align}
H(\rho_K) \geq n_{\eps,\delta} 2\delta \log(2).    
\end{align}
Choosing $\delta=1/4$ we find
\begin{align}
    H(\rho_K) \geq \left\lfloor \frac{1}{4\eps}\right\rfloor \log(2).
\end{align}

\section{Proof of \cref{lemma:PS-trick}}
    By a slight extension of the proof of \cite[Lemma 4.2]{powers1970free_states}, if $\delta \1 \leq A$ we have
    \begin{align}
        2^{x}\delta^{1/2} \norm{A^{1/2} - B^{1/2}}_2 \leq \norm{A-B}_2,
    \end{align}
    where $x=1$ if also $\delta \1\leq B$ and $x=0$ otherwise.
    Since $A\leq (1-\delta)\1$  can be rewritten as $\1-A\geq \delta \1$, we also get
    \begin{align}
       2^x\delta^{1/2} &\norm{(\1-A)^{1/2}-(\1-B)^{1/2}}_2 \leq \norm{(\1-A)-(\1-B)}_2 \nonumber = \norm{A-B}_2,
    \end{align}
    where again $x=1$ if also $B\leq (1-\delta)\1$ and $x=0$ otherwise.
    We thus find
    \begin{align}
        \norm{(\1-A)^{1/2}B^{1/2} - A^{1/2}(\1-B)^{1/2}}_2 
        &\leq \norm{(\1-A)^{1/2}(B^{1/2} - A^{1/2})}_2  + \norm{A^{1/2}\big((\1-A)^{1/2}-(\1-B)^{1/2}\big)}_2 \nonumber \\
        &\leq \norm{(\1-A)^{1/2}}\norm{B^{1/2}-A^{1/2}}_2 + \norm{A^{1/2}}\norm{(\1-A)^{1/2} - (\1-B)^{1/2}}_2 \nonumber\\
        &\leq \frac{2^{1-x}}{\delta^{1/2}}\norm{A-B}_2,
    \end{align}
    where $x=1$ if $\delta\1 \leq B\leq (1-\delta)\1$ and $x=0$ otherwise and  where we used H\"older's inequality for the last inequality.

\section{Proof of \Cref{lemma:list-sort}}

In this section, we show \cref{lemma:list-sort} of the main text.
We restate it for completeness:
\begin{lemma}\label{lemma} 
	Let $k\in \RR^n, g,f\in \RR^d$ be non-increasingly ordered vectors (e.g. $k_1 \geq k_2\geq\cdots$), such that $k_1 \geq g_1,f_1$ and $k_m \leq g_d,f_d$. Then
	$$	\norm{(k\oplus f)^\downarrow - (k\oplus g)^\downarrow}_\infty \leq \max_{|i-j| \leq d} |k_i - k_j| \leq d\,\max_i\, (k_i - k_{i+1})$$
\end{lemma}

Before we come to the proof, we note the following consequence:

\begin{corollary}
	Let $k,g,f$ be as in \cref{lemma} and define $k^{\oplus d} = (k\oplus k \oplus \cdots \oplus k)^\downarrow$.
	Then
	\begin{align}
	\norm{(k^{\oplus d} \oplus f)^\downarrow - (k^{\oplus d} \oplus g)^\downarrow}_\infty \leq \max_i\, (k_i - k_{i+1}).
	\end{align}
	\begin{proof}
		$\max_{|i-j| \leq d} |k^{\oplus d}_i - k^{\oplus d}_j| \leq \max_{|i-j|\leq 1} |k_i - k_j| \leq \max_i\, (k_i-k_{i+1})$. 
	\end{proof}
\end{corollary}

\begin{proof}[Proof of \cref{lemma}]
To simplify notation, we set $a := (k\oplus f)^\downarrow, b := (k\oplus g)^\downarrow$.
To every entry $a_x$ we associate a unique element of $f$ or $k$, indicated by $a_x \ptr f_i$ or $a_x \ptr k_i$ (similarly for $b$).
To deal with degeneracies, we make the following conventions:
\begin{itemize}
	\item If $a_x \ptr f_i, a_y \ptr f_j$ and $a_x=a_y$, with $x < y$, then $i<j$. The same rule holds for $k$-entries. ("$f$'s and $k$'s keep their ordering")
	\item If $a_x \ptr f_i, a_y \ptr k_j$ and $a_x = a_y$ then $x<y$. ("$f$'s go first").
\end{itemize}
We also use the notation $m(k\geq \lambda) = |\{ j : k_j \geq \lambda \}|$ and similarly for other relations and $f,g$.
Since the entries of $k$ are sorted, $m(k>\lambda)$ is simply the largest value of $j$ such that $k_j >\lambda$.
If a value $\lambda$ in $k$ has multiplicity $l$, so that $\lambda = k_{i} = k_{i+1}=\cdots = k_{i+l-1}$, then 
$m(f\geq k_i) = m(f \geq k_{i+1}) =\dots= m(f\geq k_{i+l-1})$.

The following consequences follow from these conventions:
\begin{enumerate}[label=(\alph*)]
	\item If $a_x\ptr f_i$ and $a_y\ptr k_j$ for $y<x$, then $k_j > f_i$. Hence, if $a_x\ptr f_i$ we have $x = i + m(k > f_i)$.
		In particular we have
		\begin{align}
			k_{x-i+1} \leq a_x = f_i < k_{m(k>f_i)} =  k_{x-i},
		\end{align}
		where we set $k_0 := k_1$ and $k_{m+1} := k_m$ (note that $k_0 \geq f_j,g_j$ and $k_{m+1} \leq f_j,g_j$ for all $j$ by our assumptions).
	\item If $a_x \ptr k_i$ and $a_y \ptr f_j$ for $y<x$, then $f_j \geq k_i$. Therefore, $x = i + m(f\geq k_i)$.
\end{enumerate}
	We now choose some $x$ and assume (without loss of generality) that $0\leq a_x- b_x$. We need to show that we can always bound
	\begin{align}
		a_x - b_x \leq k_i - k_j,\quad |i-j| = j-i \leq d.
	\end{align}
	Using the conventions from above, we distinguish different cases, depending on whether $a_x,b_x$ are $k$-valued or not:
	\begin{enumerate}
		\item $a_x \ptr k_i, b_x \ptr k_j$: We have $x = i + m(f\geq k_i) = j + m(g\geq k_j)$ with $j\geq i$. Hence,
			\begin{align}
				|i-j| =	j - i = m(f\geq k_i) - m(g\geq k_j) \leq m(f\geq k_i) \leq d. 
			\end{align}

		\item $a_x \ptr k_i, b_x \ptr g_j$: We have $x = i+m(f\geq k_i)$ and $g_j \geq k_{x-j+1}$. Hence,
			\begin{align}
				a_x - b_x \leq k_i - k_{x-j+1},
			\end{align}
			and (since $j\geq 1$)
			\begin{align}
				|i - (x-j+1)| = (x-j+1)-\big(x - m(f\geq k_i)\big)  \leq m(f\geq k_i) \leq d.
			\end{align}
		\item $a_x \ptr f_i, b_x \ptr k_j$: We have $f_i \leq k_{x-i}$ and $x = j + m(g\geq k_j)$. Hence,
			\begin{align}
				a_x - b_x \leq k_{x-i} - k_{x - m(g\geq k_j)},
			\end{align}
			with
			\begin{align}
				|x-i - (x-m(g\geq k_j))| =  i - m(g\geq k_j) \leq d.
			\end{align}

		\item$a_x \ptr f_i, b_x \ptr g_j$: We have $f_i \leq k_{x-i}$ and $g_j \geq k_{x-j+1}$. Hence,
			\begin{align}
				a_x - b_x \leq k_{x-i} - k_{x-j+1},
			\end{align}
			and (since $j\geq 1$)
			\begin{align}
				|x-i - (x-j+1)| = i - (j-1) \leq d.
			\end{align}
	\end{enumerate}
	This shows the first inequality of \cref{lemma}. The second inequality is immediate from the fact that the entries of $k$ are ordered.
	
\end{proof}
\end{document}